\pdfoutput=1
\RequirePackage{ifpdf}
\ifpdf % We are running pdfTeX in pdf mode
\documentclass[pdftex]{sigma}
\else
\documentclass{sigma}
\fi

\usepackage{tikz, pgfplots}
\usetikzlibrary{decorations.markings,arrows, calc, calligraphy}

\numberwithin{equation}{section}

\newtheorem{Theorem}{Theorem}[section]
\newtheorem{Lemma}[Theorem]{Lemma}
\newtheorem{Proposition}[Theorem]{Proposition}
 { \theoremstyle{definition}
\newtheorem{Remark}[Theorem]{Remark} }

\def \res#1{\mathop {\mathrm{res}}_{#1}}

\def \ds{\displaystyle}

\def \tt { {\un t}}
\def \s{\sigma}

\def \ddz{ \frac {\d z}{2i\pi}}
\def \ddw{ \frac {\d w}{2i\pi}}
\def\wt{\widetilde}

\def\D{\mathbb D}

\def\un {\underline}

\def \G{\Gamma}

\def \pa{\partial}
\def\C{{\mathbb C}}

\def\N{{\mathbb N}}

\def\wh{\widehat}
\def\H{{\cal H}}
\def\Z{{\mathbb Z}}

\def\d{\mathrm d}

\def\1{{\bf 1}}
\def \t{\mathbf t}

\begin{document}

\allowdisplaybreaks

\newcommand{\arXivNumber}{1703.00046}

\renewcommand{\PaperNumber}{046}

\FirstPageHeading

\ShortArticleName{The Malgrange Form and Fredholm Determinants}

\ArticleName{The Malgrange Form and Fredholm Determinants}

\Author{Marco BERTOLA~$^{\dag\ddag}$}

\AuthorNameForHeading{M.~Bertola}

\Address{$^\dag$~Department of Mathematics and Statistics, Concordia University, Montr\'eal, Canada}
\EmailD{\href{mailto:marco.bertola@concordia.ca}{marco.bertola@concordia.ca}}

\Address{$^\ddag$~Area of Mathematics SISSA/ISAS, Trieste, Italy}
\EmailD{\href{mailto:marco.bertola@sissa.it}{marco.bertola@sissa.it}}

\ArticleDates{Received March 12, 2017, in f\/inal form June 17, 2017; Published online June 22, 2017}

\Abstract{We consider the factorization problem of matrix symbols relative to a closed contour, i.e., a Riemann--Hilbert problem, where the symbol depends analytically on parameters. We show how to def\/ine a function $\tau$ which is locally analytic on the space of deformations and that is expressed as a Fredholm determinant of an operator of ``integrable'' type in the sense of Its--Izergin--Korepin--Slavnov. The construction is not unique and the non-uniqueness highlights the fact that the tau function is really the section of a~line bundle.}


\Classification{35Q15; 47A53; 47A68}

\section{Introduction}
We shall consider the following prototypical matrix Riemann--Hilbert problem (RHP) on the unit circle~$\Sigma$ (or any smooth closed simple contour):
\begin{gather}
\G_+(z;\t) = \G_-(z;\t) M(z;\t),\qquad \forall\, z\in \Sigma, \qquad \G(\infty)=\1.\label{RHP1}
\end{gather}
Here $\t$ stands for a vector of parameters which we refer to as ``deformation parameters''. The assumptions are the following;
\begin{enumerate}\itemsep=0pt
\item The matrix $M(z;\t)\in {\rm GL}_n(\C)$ is jointly analytic for $z$ in a f\/ixed tubular neighbourhood~$N(\Sigma)$ of~$\Sigma$ and~$\t$ in an open connected domain $\mathcal S$, which we refer to as the ``deformation space''.
\item The index of $\det M(z;\t)$ around $\Sigma$ vanishes for all $\t\in \mathcal S$.
\item The partial indices are generically zero, i.e., the RHP~\eqref{RHP1} generically admits solution.
\end{enumerate}
Let us remind the reader of some facts that can be extracted from \cite{ClanceyGohberg}
\begin{enumerate}\itemsep=0pt
\item [--]There exists a matrix function $Y_-(z)$ analytic and analytically invertible in $\operatorname{Ext}(\Sigma) \cup N(\Sigma)$ (and uniformly bounded) and similarly a matrix function $Y_+(z)$ analytic and analytically invertible in $\operatorname{Int}(\Sigma) \cup N(\Sigma)$ and $n$ integers $k_1,\dots, k_n$ (called {\em partial indices}) such that
\begin{gather*}
D(z) Y_+(z) = Y_-(z) M(z), \quad z\in \Sigma , \qquad D = \operatorname{diag}\big( z^{k_1}, \dots, z^{k_n}\big) .
\end{gather*}
\item [--] The RHP \eqref{RHP1} is solvable if and only if all partial indices vanish, $k_j=0$, $\forall\, j=1,\dots, n$.
\end{enumerate}
 Note that since $\operatorname{ind}_\Sigma \det M=\sum\limits_{j=1}^n k_j$, the condition (2) in our assumptions is necessary for the solvability of~\eqref{RHP1}.

We shall denote by $\D_\pm$ the interior ($+$) and the exterior ($-$) regions separated by~$\Sigma$. Def\/i\-ne~$\H_+$ to be the space of functions that are in $L^2(\Sigma, |\d z|)$ and extend to analytic functions in the interior. We will use the notation $\vec \H_+ = \H_+\otimes \C^r$ (i.e., vector-valued such functions). The vectors will be thought of as row-vectors. We also introduce the Cauchy projectors \smash{$C_\pm\colon L^2(\Sigma,|\d z|) \to \H_{\pm}$}:
\begin{gather*}
C_\pm[f](z) = \oint_{\Sigma } \frac { f(w)\d w}{(w-z)2i\pi},\qquad z \in \D_\pm.
\end{gather*}

It is well known \cite{Malgrange:IsoDef1} that the RHP \eqref{RHP1} is solvable if and only if the Toeplitz operator\footnote{Due to our choices of symbols, the matrix symbol of the relevant Toeplitz operator is~$M^{-1}$. We apologize for the inconvenience.}
\begin{gather*}
T_{S}\colon \ \vec\H_+\to \vec \H_+,\\
T_{S} [\vec f] = C_+[\vec f S] ,\qquad S(z;\t) := M^{-1}(z;\t)
\end{gather*}
is invertible, in which case the inverse is given by
\begin{gather*}
T_{S}^{-1}[\vec f] = C_+ \big[\vec f \G_-^{-1}\big]\G_+.
\end{gather*}
 Moreover the operator is Fredholm and
\begin{gather*}
\dim \ker (T_S) - \dim \operatorname{coker} (T_S) = \operatorname{ind}_\Sigma \det M =0.
\end{gather*}

There is no reasonable way, however, to def\/ine a ``determinant'' of $T_S$ as it stands. Such a~function of~$\t$ would desirably have the property that the RHP \eqref{RHP1} is not solvable if and only if this putative determinant is zero.

While this is notoriously impossible in this naive form, we now propose a proxy for the notion of determinant, in terms of a simple Fredholm determinant.

{\bf The Malgrange one-form.}
As Malgrange explains \cite{Malgrange:Deformations}
one can def\/ine a central extension on the loop group $\mathcal G:= \{M\colon \Sigma \to {\rm GL}_n(\C)\colon \operatorname{ind}_{\Sigma} \det M=0\}$ given by $ \wh {\mathcal G} =\{ (M,u)\in \mathcal G \times \C^\times\}$ with the group law\footnote{We are being a bit cavalier in this description; we invite the reader to read pp.~1373--1374 in~\cite{Malgrange:Deformations}.}
\begin{gather*}
(M, u) \cdot (\wt M, \wt u) = \big (M \wt M, u \wt u c\big(M,\wt M\big)\big) ,\qquad c(M, \wt M):= \det_{\H_+} \big(T_{M^{-1}} T_{\wt M^{-1}} T^{-1}_{(M \wt M)^{-1}}\big).
\end{gather*}
The operator in the determinant is of the form $\operatorname{Id}_{\H_+} + $ (trace class) and hence the Fredholm determinant is well def\/ined. This group law is only valid for pairs~$M$,~$\wt M$ for which the inverse of~$T^{-1}_{(M \wt M)^{-1}}$ exists. The left-invariant Maurer--Cartan form of this central extension is then given by $\big(S^{-1}\delta S, \frac {\d u}u +\wh \omega_M\big)$, where
\begin{gather*}
\wh \omega_{M} := \operatorname{Tr}_{\H_+} \big( T_S^{-1} \circ T_{\delta S} - T_{S^{-1}\delta S} \big),\qquad S:= M^{-1},
\end{gather*}
and this can be written as the following integral
\begin{gather}\label{Mal1}
\wh \omega_M = \oint_{\Sigma} \operatorname{Tr} \big(\G_+^{-1} \G_+' M^{-1} \delta M\big) \ddz.
\end{gather}
Here, and below, $\delta$ denotes the exterior total dif\/ferentiation in the deformation space $\mathcal S$:
\begin{gather*}
\delta = \sum \delta t_j \frac {\pa }{\pa t_j}.
\end{gather*}

The Malgrange form is a {\em logarithmic} form in the sense that it has only simple poles on a co-dimension~$1$ analytic submanifold of the deformation space~$\mathcal S$ and with positive integer Poincar\'e residue along it; this manifold is precisely the exceptional ``divisor'' $(\Theta)\subset \mathcal S$ (the {\em Malgrange divisor}) where the RHP~\eqref{RHP1} becomes non solvable, i.e., where some partial indices of the Birkhof\/f factorization become non-zero.

Closely related to \eqref{Mal1} is the following one-form, which we still name after Malgrange:
\begin{gather}\label{Mal2}
\omega_M := \oint_{\Sigma} \operatorname{Tr} \big(\G_-^{-1} \G_-' \delta M M^{-1}\big) \ddz.
\end{gather}
It is also a logarithmic form with the same pole-divisor; indeed one verif\/ies that
\begin{gather*}
\wh \omega_M - \omega_M = \oint_\Sigma \operatorname{Tr} \big(M' M^{-1} \delta M M^{-1}\big) \ddz,
\end{gather*}
which is an analytic form of the deformation parameters $\t\in \mathcal S$. In \cite{BertolaIsoTau} the one-form~\eqref{Mal2} was posited as an object of interest for general Riemann--Hilbert problems (not necessarily on closed contours) and its exterior derivative computed (with an important correction in~\cite{BertolaCorrection}, which is however irrelevant in the present context). It was computed (but the computation can be traced back to Malgrange himself in this case) that
\begin{gather}\label{deltaomega}
\delta \omega_M = \frac 1 2 \oint_{\Sigma} \operatorname{Tr} \left(\Xi(z) \wedge \frac {\d}{\d z} \Xi(z) \right)\ddz, \qquad \Xi(z;\t):= \delta M(z;\t) M(z;\t)^{-1}.
\end{gather}
It appears from \eqref{deltaomega} that this two form $\delta \omega_M$ is not only closed, but also smooth on the whole of~$\mathcal S$, including the Malgrange-divisor. As such, it def\/ines a line-bundle $\mathcal L$ over $\mathcal S$ by the usual construction: one covers $\mathcal S$ by appropriate open sets $U_\alpha$ where $\delta \omega_M =\delta \theta_\alpha$; on the overlap $U_\alpha \cap U_\beta$ the form $\theta_{\alpha }-\theta_{\beta}$ is also exact and one def\/ines then the transition functions by $g_{\alpha\beta}(\t) = \exp \big(\int \theta_{\alpha}-\theta_\beta\big)$. Then a section of this line bundle is provided by the collection of functions $\tau_\alpha\colon U_\alpha \to \C$ such that
\begin{gather*}
\tau_{\alpha}(\t) =\exp\left[\int ( \omega_M - \theta_\alpha )\right].
\end{gather*}
Since $\omega_M$ is a logarithmic form and each $\theta_\alpha$ is analytic in the respective~$U_\alpha$, the func\-tions~$\tau_\alpha(\t)$ have zero of f\/inite order precisely on the Malgrange divisor $(\Theta)\subset \mathcal S$ (under appropriate transversality assumptions, the order of the zero is the dimension of $\operatorname{Ker} T_S$).

Our goal is to provide an explicit construction of the $\tau_\alpha$'s in terms of Fredholm determinants of simple operators of the Its--Izergin--Korepin--Slavnov ``integrable'' type~\cite{ItsIzerginKorepinSlavnov}. Their def\/inition is recalled in due time.

\section{Construction of the Fredholm determinants}

The construction carried out below is not unique, and also only local in the deformation space~$\mathcal S$; this is however not only not a problem, but rather an interesting feature, as we will illustrate in the case of ${\rm SL}_2(\C)$. The non-uniqueness is precisely a consequence of the fact that we are trying to compute a section of the aforementioned line bundle.

{\bf Preparatory step.} The assumption that $M(z;\t)\in {\rm GL}_n$ can be replaced without loss of generality with $M(z;\t)\in {\rm SL}_n$; this is so because of the assumption on the index of $\det M$. Indeed we can solve the scalar problem $y_+ (z) =y_-(z) \det M(z)$, $y(\infty)=1$ and then def\/ine a~new RHP where $\wt \G_\pm (z) := \G_\pm \operatorname{diag}\big(y_\pm^{-1} , 1,\dots\big)$ and hence the new matrix jump for $\wt \G$ is $\wt M (z)= \operatorname{diag}(y_-(z) , 1,\dots) M(z) \operatorname{diag}\big(y_+^{-1}(z) , 1,\dots\big) $ with $\det \wt M \equiv 1$. For this reason, from here on we assume $M\in {\rm SL}_n(\C)$.

We def\/ine an {\it elementary} matrix (for our purposes) to be a matrix of the form $\1 + cE_{jk}$, with $j\neq k$, where $E_{jk}$ denotes the $(j,k)$-unit matrix.

\begin{Lemma}\label{lemmafactor}
Any matrix $M\in {\rm SL}_n(\C)$ can be written as a product of elementary matrices. The entries of the factorization are rational in the entries of~$M$ with denominators that are monomials in a suitable set of $n-1$ nested minors of~$M$.
\end{Lemma}

\begin{proof}
We recall that given any matrix $M\in {\rm SL}_n$, there is a permutation $\Pi$ of the columns such that the principal minors (the determinants of the top left square submatrices) do not vanish, and hence we can write it as
\begin{gather*}
M = L D U \Pi = \wh M\Pi,
\end{gather*}
where $L$, $U$ are lower/upper triangular matrices with unit on the diagonal and $D = \operatorname{diag}(x_1,\dots$, $x_n)$ is a diagonal matrix (see for example \cite[Vol.~1, Chapter~II]{Gantmacher}). Denote $q_\ell = \det \big[\wh M_{j,k}\big]_{j,k\leq \ell}$ the principal minors; these are the nested minors of the original matrix $M$ alluded to in the statement. The matrices~$L$,~$U$ are rational in the entries of~$M$ and with denominators that are monomials in the~$q_\ell$'s.

Now, both $L$, $U$ can clearly be written as products of elementary matrices whose coef\/f\/icients are polynomials in the entries of~$L$,~$U$ (respectively) and so it remains to show that we can write~$D$ as product of elementary matrices.

To this end we observe the 'LULU' identity (there is a similar `ULUL' identity)
\begin{gather*}
\left[
\begin{matrix}
x & 0\\
0 &\frac 1 x
\end{matrix}
\right] =
\left[
\begin{matrix}
1 & 0\\
\frac {1-x}x & 1
\end{matrix}
\right]
\left[
\begin{matrix}
1 & 1\\
0 & 1
\end{matrix}
\right]
\left[
\begin{matrix}
1 & 0\\
x-1 & 1
\end{matrix}
\right]
\left[
\begin{matrix}
1 & -\frac 1 x \\
0 & 1
\end{matrix}
\right].
\end{gather*}
As $D\in {\rm SL}_n$, we can represent it in terms of product of embedded ${\rm SL}_2$ matrices using the root decomposition of ${\rm SL}_n$:
\begin{gather*}
D = \operatorname{diag}\left(x_1,\frac 1 x_1,1,\dots,1\right)
\operatorname{diag}\left(1,x_2 x_1, \frac 1{x_1 x_2},1,\dots,1\right)\\
\hphantom{D =}{}\times
\operatorname{diag}\left(1,1, x_1x_2x_3, \frac 1{x_1 x_2x_3},1,\dots,1\right) \cdots,
\end{gather*}
and then embed the LULU identity for each factor.

Finally, also permutation matrices can be written as product of elementary matrices embedding appropriately the simple identity;
\begin{gather*}
\left[
\begin{matrix}
0 & -1 \\ 1 & 0
\end{matrix}
\right]
=
\left[
\begin{matrix}
1 & 0 \\ 1 & 1
\end{matrix}
\right]
\left[
\begin{matrix}
1 & -1 \\ 0 & 1
\end{matrix}
\right]
\left[
\begin{matrix}
1 &0 \\ 1 & 1
\end{matrix}
\right].
\end{gather*}
This concludes the proof.
\end{proof}

Let now $M(z)$ be an ${\rm SL}_n$ matrix valued function, analytic in a tubular neighbourhood $N(\Sigma)$ of~$\Sigma$, and let~$q_\ell(z)$, $\ell=1,\dots, n-1$ be the nested minors alluded to in the Lemma so that they are not identically zero. Since the entries are analytic in~$N(\Sigma)$ we can slightly deform the contour~$\Sigma$ to a contour~$\wt \Sigma$ that avoids all zeroes of every principal minor~$q_\ell(z)$. The resulting RHP is ``equivalent'' to the original in the sense that the solvability of one implies the solvability of the other. Note also that this deformation can be done in a piecewise constant way locally with respect to $\t\in \mathcal S$. Thus we have
\begin{gather*}
M(z) = F_1(z) \cdots F_R(z),\qquad F_\nu(z) = \1 + a_\nu (z) E_{j_\nu,k_\nu} ,\qquad \nu = 1,\dots, R.
\end{gather*}

\begin{figure}[t]\centering

\begin{minipage}{0.44\textwidth}\baselineskip 16pt plus 1pt minus 1pt
\resizebox{0.9\textwidth}{!}{
\begin{tikzpicture}[scale=2]

\fill [line width=0, fill=gray!30!white, fill opacity = 1] (0,0) circle[radius=2.2];
\fill [line width=0, fill=gray!30!red, fill opacity = 0.2] (0,0) circle[radius=2];
\fill [line width=0, fill=gray!30!blue, fill opacity = 0.2] (0,0) circle[radius=1.7];
\fill [line width=0, fill=gray!30!green, fill opacity = 0.2] (0,0) circle[radius=1.4];
\fill [line width=0, fill=gray!30!white, fill opacity = 1] (0,0) circle[radius=1.1];
\fill [line width=0, fill=white] (0,0) circle[radius=1.0];
\draw [dotted] circle[radius=2.2];
\draw [dotted] circle[radius=1.0];

\draw [ postaction={decorate,decoration={markings,mark=at position 0.2 with {\arrow[line width=1.5pt]{>}}}}]
(2,0) arc [radius=2, start angle = 0, end angle =360] node [pos=0.23,sloped] {$\Sigma = \Sigma_1$}
node [pos=0.1,sloped, below] {$\D_1$};

\draw [ postaction={decorate,decoration={markings,mark=at position 0.3 with {\arrow[line width=1.5pt]{>}}}}]
(1.7,0) arc [radius=1.7, start angle = 0, end angle =360] node [pos=0.25,sloped] {$\Sigma_2$}
node [pos=0.15,sloped, below] {$\D_2$};

\draw [ postaction={decorate,decoration={markings,mark=at position 0.4 with {\arrow[line width=1.5pt]{>}}}}]
(1.4,0) arc [radius=1.4, start angle = 0, end angle =360] node [pos=0.3,sloped] {$\Sigma_3$}
node [pos=0.2,sloped, below] {$\D_{R-1}$};

\draw [ postaction={decorate,decoration={markings,mark=at position 0.5 with {\arrow[line width=1.5pt]{>}}}}]
(1.1,0) arc [radius=1.1, start angle = 0, end angle =360] node [pos=0.35,sloped] {$\Sigma_R$};
\node at (0.3,0) {$\D_R=\D_+$};
\node at (2,1.6) {$\D_0=\D_-$};
\end{tikzpicture}}\end{minipage}
\caption{An illustration of the splitting of the jump matrix into elementary jumps. The matrix $M(z;\t)$ is analytic in the shaded regions.}\label{split}
\end{figure}
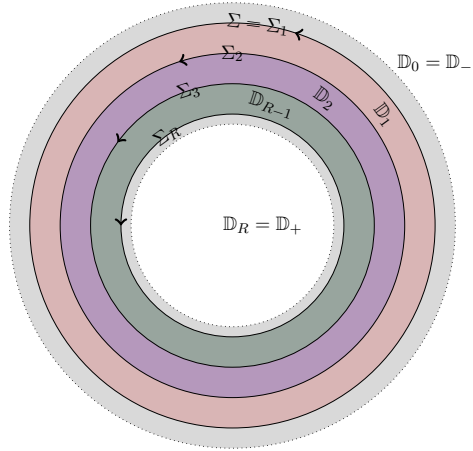

Corresponding to this factorization we can def\/ine an equivalent RHP with jumps on $R$ contours $\Sigma_1,\dots, \Sigma_R$, with $\Sigma_1=\Sigma$ and $\Sigma_{j+1}$ in the interior of $\Sigma_j$ and all of them in the joint domain of analyticity of the scalar functions $a_\nu(z)$ (see Fig.~\ref{split}) which may have poles only at the zeroes of the principal minors $q_\ell(z)$ of $M(z)$. This is accomplished by ``extending'' the matrix $\G_-(z)$ to the annular regions $\D_0 = \D_-$ and $ \D_j = \operatorname{Int}(\Sigma_j)\cap \operatorname{Ext}(\Sigma_{j+1})$ as
\begin{gather}
 \Theta_0 (z) :=\G_-(z), \qquad \forall\, z\in \D_0,\nonumber \\
 \Theta_\nu(z) := \G_-(z) \overrightarrow{\prod_{\ell=1}^{\nu} } F_\ell(z) ,\qquad \forall\, z\in \D_\nu.\label{Thetanus}
\end{gather}
By doing so we obtain the following relations
\begin{gather*}
\Theta_{\nu}(z) = \Theta_{\nu-1}(z) F_{\nu}(z) ,\qquad \forall\, z\in \Sigma_{\nu} ,\qquad \nu=1,\dots, R.
\end{gather*}
The piecewise analytic matrix function $\Theta(z)$ whose restriction to $\D_\nu$ coincides with the matri\-ces~$\Theta_\nu$~\eqref{Thetanus}, satisf\/ies a f\/inal RHP
 \begin{gather}\label{RHPTheta}
\Theta_+(z) = \Theta_-(z) F_\nu (z),\qquad \forall\, z\in \Sigma_\nu,\qquad \Theta(\infty)=\1.
\end{gather}

This type of RHP is of the general type of ``integrable kernels'' and its solvability can be determined by computing the Fredholm determinant of an integral operator of $L^2(\bigsqcup \Sigma_\nu, |\d z|) \simeq \bigoplus_{\nu=1}^R L^2(\Sigma_\nu, |\d z|) $ with kernel (we use the same symbol for the operator and its kernel)
\begin{gather}
K(z,w) = \frac {\vec f^T(z) \vec g(w)}{2i\pi(w-z)} ,\qquad\vec f(z) = \sum_{\nu=1}^R {\bf e}_{j_\nu} \chi_\nu(z) a_\nu(z) ,\qquad
\vec g(z) = \sum_{\nu=1}^R {\bf e}_{k_\nu} \chi_\nu(z),\label{K}
\end{gather}
where $\chi_\nu(z)$ is the projector (indicator function) on the component $L^2(\Sigma_\nu, |\d z|)$. Indeed, as explained in \cite{BertolaCafasso1, ItsHarnad, ItsIzerginKorepinSlavnov}, the Fredholm determinant $\det (\operatorname{Id} - K)$ is zero if and only if the RHP~\eqref{RHPTheta} is non-solvable and moreover the resolvent operator $R = K(\operatorname{Id}-K)^{-1}$ of $K$ has kernel
\begin{gather*}
R(z,w) =\frac{ \vec f^t(z) \Theta^t(z) (\Theta^t)^{-1} (w)\vec g(w)}{z-w}.
\end{gather*}
\begin{Theorem}
The RHP \eqref{RHPTheta} and hence \eqref{RHP1} is solvable if and only if $\tau := \det (\operatorname{Id} - K) \neq 0$.
\end{Theorem}

\begin{Proposition}[{see, e.g., \cite[Theorem~2.1]{BertolaCafasso1}}] \label{prop11}
Let $\pa$ be any deformation of the functions $a_\nu(z)$, then
\begin{gather*}
\pa \ln \tau = \sum_{\nu=1}^R \oint_{\Sigma_\nu} \operatorname{Tr} \big(\Theta_-^{-1} \Theta_-'(z) \pa F_\nu F_\nu^{-1} \big) \frac {\d z}{2i\pi}
 = \sum_{\nu=1}^R \oint_{\Sigma_\nu} \big(\Theta_{\nu-1}^{-1} \Theta_{\nu-1}'(z) \big)_{k_\nu, j_\nu} \pa a_\nu(z) \frac {\d z}{2i\pi}.
\end{gather*}
\end{Proposition}

\section[The ${\rm SL}_2$ case]{The $\boldsymbol{{\rm SL}_2}$ case}

We would like to express the Malgrange one-form directly in terms of the $\tau$ function (Fredholm determinant). Rather than obscuring the simple idea with the general case, we consider in detail the ${\rm SL}_2$ case. Let $M(z;\t)$ be analytic in $(z,\t)\in N(\Sigma)\times \mathcal S$ and with values in ${\rm SL}_2(\C)$.
Using the general scheme above, we have the following factorizations
\begin{gather}
 {(1)}\ \ \left[
\begin{matrix}
a & b\\
c &d
\end{matrix}
\right] =
\underbrace{\left[
\begin{matrix}
1 & 0\\
\frac {1+c-a} a & 1
\end{matrix}
\right]}_{F_1(z)}
\underbrace{
\left[
\begin{matrix}
1 & 1\\
0 & 1
\end{matrix}
\right]}_{F_2(z)}
\underbrace{
\left[
\begin{matrix}
1 & 0\\
a-1 & 1
\end{matrix}
\right]}_{F_3(z)}
\underbrace{
\left[
\begin{matrix}
1 & \frac {b-1} a\\
0 & 1
\end{matrix}
\right]}_{F_4(z)} ,\qquad a \not\equiv 0,\nonumber\\
{(2)}\ \ \left[
\begin{matrix}
a & b\\
c &d
\end{matrix}
\right] =
\underbrace{\left[
\begin{matrix}
1 & \frac {1+b-d} d\\
0 & 1
\end{matrix}
\right]}_{F_1(z)}
\underbrace{
\left[
\begin{matrix}
1 & 0\\
1 & 1
\end{matrix}
\right]}_{F_2(z)}
\underbrace{
\left[
\begin{matrix}
1 & d-1\\
0 & 1
\end{matrix}
\right]}_{F_3(z)}
\underbrace{
\left[
\begin{matrix}
1 & 0\\
\frac {c-1}d & 1
\end{matrix}
\right]}_{F_4(z)}, \qquad d \not\equiv 0,\nonumber\\
{(3)}\ \ \left[
\begin{matrix}
a & b\\
-\frac 1 b &0
\end{matrix}
\right] =
\underbrace{\left[
\begin{matrix}
1 & b-ab\\
0 & 1
\end{matrix}
\right]}_{F_1(z)}
\underbrace{
\left[
\begin{matrix}
1 &0 \\
 -\frac 1 b & 1
\end{matrix}
\right]}_{F_2(z)}
\underbrace{
\left[
\begin{matrix}
1 & b\\
0 & 1
\end{matrix}
\right]}_{F_3(z)},\nonumber\\
 {(4)}\ \ \left[
\begin{matrix}
0 & b\\
-\frac 1 b &d
\end{matrix}
\right] =
\underbrace{\left[
\begin{matrix}
1 &\frac {d-1}b\\
0 & 1
\end{matrix}
\right]}_{F_1(z)}
\underbrace{
\left[
\begin{matrix}
1 &0 \\
 b & 1
\end{matrix}
\right]}_{F_2(z)}
\underbrace{
\left[
\begin{matrix}
1 & -\frac 1 b \\
0 & 1
\end{matrix}
\right]}_{F_3(z)},\nonumber \\
{(5)}\ \ \left[
\begin{matrix}
0 & b\\
-\frac 1 b &0
\end{matrix}
\right] =
\underbrace{\left[
\begin{matrix}
1 & 0\\
b & 1
\end{matrix}
\right]}_{F_1(z)}
\underbrace{
\left[
\begin{matrix}
1 & -\frac 1 b \\
0 & 1
\end{matrix}
\right]}_{F_2(z)}
\underbrace{
\left[
\begin{matrix}
1 & 0\\
b & 1
\end{matrix}
\right]}_{F_3(z)}.\label{Factorizations}
\end{gather}

\subsection{Fredholm determinants for dif\/ferent factorizations}
Each of the factorization \eqref{Factorizations} leads to an integrable operator of the form \eqref{K} and hence to a~corresponding Fredholm determinant; we now establish their mutual relationships.

There are two types of questions that we address here
\begin{enumerate}\itemsep=0pt
\item How are the Fredholm determinants associated with the dif\/ferent factorizations $(1,2)$ in~\eqref{Factorizations} related to each other?
\item For a f\/ixed factorization, how does the Fredholm determinant depend on the choice of contour $\Sigma$ (within the analyticity domain $N(\Sigma)$).
\end{enumerate}

Consider the cases \eqref{Factorizations}$_{(\rho)}$, $\rho=1,\dots, 5$.
We compute the logarithmic derivative of the corresponding Fredholm determinant using Proposition~\ref{prop11}
\begin{gather}
\pa \ln \tau_{(\rho)} = \sum_{\nu=1}^R \oint_{\Sigma_\nu} \operatorname{Tr} \big(\Theta_-^{-1} \Theta_-'\pa F_\nu F_\nu^{-1} \big) \frac {\d z}{2i\pi}.\label{42}
\end{gather}
By using the relationship \eqref{Thetanus} between $\Gamma_-$ (extended to an analytic function on $\operatorname{Ext}(\Sigma) \cup N(\Sigma)$), we can re-express it in terms of the Malgrange one-form of the original problem~\eqref{RHP1}; we use the fact that we can deform the contours back to $\Sigma = \Sigma_1$ by Cauchy's theorem. Plugging~\eqref{Thetanus} appropriately in~\eqref{42} and using Leibnitz rule, after a short computation we obtain {\samepage
\begin{gather}
\delta \ln \tau_{_{(\rho)}} = \oint_{\Sigma} \operatorname{Tr} \big(\G_-^{-1} \G_-' \pa M M^{-1} \big) \frac {\d z}{2i\pi} \label{omegarho} \\
\hphantom{\delta \ln \tau_{_{(\rho)}} =}{} +\! \oint_{\Sigma}\! \operatorname{Tr} \big( F_1^{-1} F_1 ' \pa F_2 F_2^{-1} +
 F_{12} ^{-1} F_{12}' \pa F_3 F_3^{-1} + F_{123}^{-1} F_{123}' \pa F_4 F_4^{-1} \big) \frac {\d z}{2i\pi} =:\omega_M + \theta_{(\rho)},\!\nonumber
 \end{gather}
 where $F_{1\dots k} = F_1F_2F_3\cdots F_k$ (if the factorization has only three term, then we set $F_4\equiv \1$).}

{\bf The cases (3,4,5).} The last cases \eqref{Factorizations}$_{{(3,4,5)}}$ lead essentially to a RHP with a triangular jump; it suf\/f\/ices to re-def\/ine~$\Gamma$ by $\Gamma \left[\begin{smallmatrix}
0 & 1\\-1 & 0
\end{smallmatrix}\right] $ for $z\in \operatorname{Int}(\Sigma)$. If the index of $b$ is zero, $\operatorname{ind}_\Sigma b=0$,then the solution can be written explicitly in closed form and it is interesting to compute the Fredholm determinant associated to our factorization of the matrix. We will show
\begin{Proposition}
In the cases~\eqref{Factorizations}$_{{(3,4,5)}}$ and under the additional assumption that $\operatorname{ind}_\Sigma b=0$, the $\tau$ function given by $\det (\operatorname{Id}_{L^2(\cup\Sigma_j)} - K)$ and $K$ as in~\eqref{K}, equals the constant in the strong Szeg\"o formula, given by {\rm \cite{BasorWidom-BO, BorodinOkounkov}}
\begin{gather*}
\tau = \exp \bigg[ \sum_{j>0} j \beta_{-j} \beta_j\bigg] = \det_{\H_+} T_b T_{b^{-1}},
\end{gather*}
 where now the Toeplitz operator is for the scalar symbol $b(z)$ and $\H_+$ is the Hardy space of scalar functions analytic in $\D_+$ and $\beta_j$ are the coefficients of $\ln b(z)$ in the Laurent expansion centered at the origin
 \begin{gather}
\beta(z):= \ln b(z) = \sum_{j\in \Z} \beta_j z^{j}.\label{Laurent}
 \end{gather}
The same applies in the case that the jump is triangular $(b\equiv 0$ and/or $c\equiv 0)$ under the assumption $\operatorname{ind}_\Sigma a=0$ and replacing $b$ with $a$ in the above formulas.
\end{Proposition}

\begin{proof}
From a direct computation of the term $\theta_{(\rho)}$ in \eqref{omegarho}, we f\/ind
\begin{gather*}
 \theta_{_{(3,4,5)}} = \oint_{\Sigma}\left( \delta \beta \frac {\d}{\d z} \beta \right) \ddz ,\qquad \beta(z):= \ln b(z)
\end{gather*}
and the solution of the RHP is explicit,
\begin{gather*}
\Gamma(z) =
\begin{cases}
\ds \left[
\begin{matrix}
 \ds {\rm e}^{B(z)}
 &- \ds {\rm e}^{-B(z) } \oint_\Sigma \frac {a(w)b(w) {\rm e}^{2B_-(w)}}{w-z} \ddw
 \\
 0 & {\rm e}^{-B(z)}
\end{matrix}
\right],
 & z\in \D_-, \\
\ds \left[
\begin{matrix}
 \ds {\rm e}^{B(z)}
 &- \ds {\rm e}^{-B(z) } \oint_\Sigma \frac {a(w)b(w) {\rm e}^{2B_-(w)}}{w-z} \ddw
 \\
 0 & {\rm e}^{-B(z)}
\end{matrix}
\right]\left[
\begin{matrix}
0 & 1\\
-1 &0
\end{matrix}
\right],
 & z\in \D_+,
\end{cases}
\\
B(z) := \oint_{\Sigma} \frac {\beta(w)}{w-z} \frac{\d w}{2i\pi}.
\end{gather*}
Thus we can write explicitly the Malgrange form:
\begin{gather*}
\omega_M = \oint_\Sigma \operatorname{Tr} \big(\G_-^{-1} \G_-' \delta M M^{-1}\big) \ddz
= \oint_\Sigma \operatorname{Tr} \left(\s_3 \oint_{\Sigma} \frac {\beta(w)}{(w-z_-)^2} \ddw\delta \beta(z) \s_3\right) \ddz \\
\hphantom{\omega_M}{} = 2\oint_{\Sigma_-} \oint_{\Sigma} \frac {\beta(w)\delta (\beta(z))}{(w-z)^2} \ddw \ddz.
\end{gather*}
In this integral, $z$ is integrated on a slightly ``larger'' contour $\Sigma_-$. The Fredholm determinant satisf\/ies
\begin{gather*}
\delta \ln \tau_{(3,4,5)} =\omega_M+ \theta_{_{(3,4,5)}}= 2\oint_{\Sigma_-} \oint_{\Sigma} \frac {\beta(w)\delta (\beta(z))}{(w-z)^2} \ddw \ddz+
 \oint_\Sigma\big[ \beta (\delta \beta)' \big] \ddw.
\end{gather*}
We now show that
\begin{gather*}
\delta \ln \tau = \delta \oint_{\Sigma_-} \oint_{\Sigma} \frac {\beta(w) \beta(z)}{(w-z)^2} \ddw \ddz=
 \oint_{\Sigma_-} \oint_{\Sigma} \frac {\delta \beta(w) \beta(z) + \beta(w) \delta \beta(z)}{(w-z)^2} \ddw \ddz.
\end{gather*}
Indeed, the exchange of order of integration of one of the addenda (and relabeling the variables) yields the other term plus the residue on the diagonal,
\begin{gather*}
 \oint_{\Sigma_-} \oint_{\Sigma} \frac {\delta \beta(w) \beta(z) + \beta(w) \delta \beta(z)}{(w-z)^2} \ddw \ddz\\
 \qquad{} =
 \oint_{\Sigma_-} \oint_{\Sigma} \frac {2 \beta(w) \delta \beta(z)}{(w-z)^2} \ddw \ddz + \oint_{\Sigma } {\beta(z) (\delta \beta)' }\ddz.
\end{gather*}
Therefore, in conclusion, we have (we f\/ix the overall constant of $\tau$ by requiring it to be $1$ for $b\equiv 1$)
\begin{gather}
\ln \tau = \oint_{\Sigma_-} \oint_{\Sigma} \frac {\ln b(w) \ln b (z)}{(w-z)^2} \ddw \ddz. \label{33}
\end{gather}
If we write a Laurent expansion of $\ln b$ \eqref{Laurent} the formula \eqref{33} gives the explicit expression
\begin{gather*}
\ln \tau = \sum_{j>0} j \beta_{-j} \beta_j,
\end{gather*}
which is also the formula for the second Szeg\"o limit theorem for the limit of the Toeplitz determinants of the symbol~$b(z)$, and it is known to to be the Fredholm determinant of an opera\-tor~\mbox{\cite{BasorWidom-BO, BorodinOkounkov}}
\begin{gather*}
\tau = \det_{\H_+} T_b T_{b^{-1}},
\end{gather*}
where now the Toeplitz operator is for the scalar symbol $b(z)$ and $\H_+$ is the Hardy space of scalar functions analytic in~$\D_+$.
\end{proof}

\begin{Remark}
The index assumption is only necessary for the case \eqref{Factorizations}$_{{(5)}}$ or \eqref{Factorizations}$_{{(1)}}$ when $b\equiv 0 \equiv c$, because in these situations the RHP separates into two scalar problems. However, the assumption $\operatorname{ind}_\Sigma b=0$ for cases (3,4) or $\operatorname{ind}_\Sigma a =0$ for the triangular case is not necessary as we now show (in the latter form). Consider the RHP
\begin{gather*}
Y_+ = Y_- \left[\begin{matrix}
1 & \mu(z)\\
0& 1
\end{matrix}\right], \quad z\in \Sigma,\qquad
Y(z) = \big(\1 + \mathcal O\big(z^{-1}\big)\big) z^{n\s_3}, \quad z\to\infty.
\end{gather*}
By def\/ining
\begin{gather*}\G(z) =
\begin{cases}
Y(z), & z\in \operatorname{Int}(\Sigma),\\
 Y(z) z^{-n\s_3}, & z\in \operatorname{Ext}(\Sigma),
\end{cases}
\end{gather*}
we are lead to the RHP in standard form
\begin{gather*}
\G_+ = \G_- \left[\begin{matrix}
z^n & z^n\mu(z)\\
0& z^{-n}
\end{matrix}\right], \quad z\in \Sigma,\qquad \G(\infty) =\1,
\end{gather*}
which is case (1) with $c\equiv 0$ (or essentially cases (3,4) up to a multiplication by piecewise constant matrices). Even if $a(z) =z^n$ and thus $\operatorname{ind}_\Sigma a = n$, this problem is still generically solvable in terms of appropriate ``orthogonal polynomials'' $\{p_\ell(z)\}_{\ell\in \N}$ def\/ined by the ``orthogonality'' property
\begin{gather*}
\oint_{\Sigma} p_\ell(z) p_k(z) \mu(z)\d z = h_\ell \delta_{\ell k}.
\end{gather*}
Then the solution of the $Y$-problem is written as
\begin{gather*}
Y(z) = \left[
\begin{matrix}
p_n(z) & \ds \oint_{\Sigma} \frac {p_n(w)\mu(w)\d w}{(w-z)2i\pi}\vspace{1mm}\\
\ds \frac{-2i\pi p_{n-1}(z)}{h_{n-1}} & \ds \oint_{\Sigma} \frac {-p_{n-1}(w)\mu(w)\d w}{(w-z) h_{n-1}}
\end{matrix}
\right],
\end{gather*}
and the solvability depends only on the condition \cite{Chihara}
\begin{gather*}
\det\left[\oint_{\Sigma} z^{\ell+j-2} \mu(z) \d z \right]_{j,\ell=1}^{n} \neq 0.
\end{gather*}
\end{Remark}

{\bf The cases (1,2).}
A straightforward computation using the explicit expression \eqref{Factorizations}$_{(1)}$ yields
\begin{gather}
\pa \ln \tau_{_{(1)}} = \oint_{\Sigma} \Big[\operatorname{Tr} \big(\G_-^{-1} \G_-' \pa M M^{-1} \big) + \pa \ln a (\ln a)' (1 + bc)\nonumber\\
\hphantom{\pa \ln \tau_{_{(1)}} = \oint_{\Sigma} \Big[}{}
 - c \pa b (\ln a)' + c' \pa b - c' b\pa(\ln a)\Big]\ddz\nonumber \\
\hphantom{\pa \ln \tau_{_{(1)}}}{}
 = \oint_{\Sigma} \operatorname{Tr} \big(\G^{-1} \G' \pa M M^{-1} \big) \frac {\d z}{2i\pi}
 + \underbrace{\oint_{\Sigma} \left( \frac {(a)'}{a}(d \pa a - c \pa b )+ { a c'} \pa \left(\frac b a\right)\right)\frac {\d z}{2i\pi}
}_{\theta_{_{(1)}}}.\label{theta}
\end{gather}
Since $\pa \ln \tau$ is a closed dif\/ferential, the exterior derivative of~$\theta_{_{(1)}}$ must be opposite to the one of the f\/irst term, which is given by \eqref{deltaomega}. Let us verify this directly; to this end we compute the exterior derivative of $\theta_{_{(1)}}$. A~straightforward computation yields
\begin{gather*}
\delta \theta_{_{(1)}}= \oint_{\Sigma}
\left(\delta b \wedge (\delta c)' + \frac{ \delta a \wedge (\delta a)' }{a^2} (1 + bc)
+\delta a \wedge \left(\frac {\delta c}a\right)' {b}\right.\\
\left. \hphantom{\delta \theta_{_{(1)}}= \oint_{\Sigma}}{}
-\delta a \wedge \delta b \frac {(c)'}a-\delta b\wedge (\delta a)' \frac {c}{a}-\delta b \wedge \delta c \frac{ a'}{a} \right)\ddz.
\end{gather*}
The exterior derivative of $\omega_M$ is given by~\eqref{deltaomega}, in which we can insert the explicit expression of~$M(z;\t)$; after a somewhat lengthy but straightforward computation we f\/ind that
\begin{gather*}
\delta \omega_M + \delta \theta_{_{(1)}} = -\frac 1 2 \oint_{\Sigma}\frac {\d}{\d z}\left(\frac { c} a \delta a \wedge \delta b
-\frac {b}{a} \delta a \wedge \delta c + \delta b\wedge \delta c \right)\ddz = 0,
\end{gather*}
thus conf\/irming that the dif\/ferential $\delta \ln \tau_{_{(1)}}$ is indeed closed (of course it must be, since the tau function is a Fredholm determinant!). The case ${(2)}$ is analogous with the replacements~$b\leftrightarrow c$ and $a\leftrightarrow d$.

\subsubsection[Determinants for dif\/ferent choices of $\Sigma$]{Determinants for dif\/ferent choices of $\boldsymbol{\Sigma}$}

The factorizations \eqref{Factorizations}$_{(1,2)}$ require that we deform the contour $\Sigma$ so that $a(z)$ (or $d(z)$) does not have any zero on~$\Sigma$. This leads to completely equivalent RHPs of the form~\eqref{RHP1} but not entirely equivalent RHP when expressed in the form~\eqref{RHPTheta}.

Note that \eqref{Factorizations}$_{(3,4,5)}$ do not suf\/fer from this ambiguity, because~$b(z)$ cannot have any zeroes in~$N(\Sigma)$ since we assumed analyticity of the jump matrix $M(z;\t)$.

Consider the factorization~\eqref{Factorizations}$_{(1)}$ (with similar considerations applying to the other factorization); in general $a(z;\t)$ has zeroes in its domain of analyticity, and their positions depend on~$\t$. Therefore it may be necessary, when considering the dependence on~$\t$, to move the contour so that certain zeroes are to the left or to the right of it because a zero may sweep across~$N(\Sigma)$ as we vary $\t\in \mathcal S$.

\begin{figure}[t]\centering
\begin{tikzpicture}[scale=2.4]
\fill [ fill=gray!30!white, fill opacity = 0.41]
plot [smooth cycle, tension =1]
 coordinates{(0:1.1) (20:0.9) (80:0.8) (190:1.1) (260: 0.4) (330:0.9)};
\fill[ fill=white, fill opacity = 1]
plot [smooth cycle, tension =1]
 coordinates{(0:0.8) (20:0.6) (80:0.5) (190:0.8) (290: 0.12) (330:0.2)};

\draw [ postaction={decorate,decoration={markings,mark=at position 0.2 with {\arrow[line width=1.5pt]{>}}}}]
plot [smooth cycle, tension =1]
 coordinates{(0:1) (20:0.8) (80:0.7) (190:1) (260: 0.3) (330:0.8)};
\draw [ postaction={decorate,decoration={markings,mark=at position 0.5 with {\arrow[line width=1.5pt]{>}}}}]
plot [smooth cycle, tension =1]
 coordinates{(0:0.9) (20:0.7) (80:0.6) (190:0.9) (260: 0.2) (330:0.4)}
;
\node at (0.73,0.5) {$\wt \Sigma$};
\node at (0.58,0.2) {$ \Sigma$};
\draw [fill] (0.2,-0.3) circle [radius=0.01];
\draw [fill] (0.2,0.65) circle [radius=0.01];
\draw [fill] (0.7,-0.2) circle [radius=0.01];
\end{tikzpicture}
\caption{The two contours and the zeroes of $a(z;\t)$ within the enclosed region. The shaded area is $N(\Sigma)$ where $M(z;\t)$ is analytic (uniformly w.r.t.~$\t\in \mathcal S$).}\label{ff1}
\end{figure}
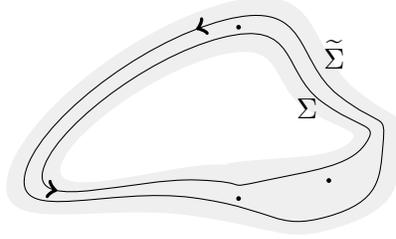

So, let $\Sigma$, $\wt \Sigma$ be two contours in the common domain of analyticity of~$M(z;\t)$ and such that~$a(z;\t)$ has no zeroes on either one and~$\Sigma \subset \operatorname{Int}(\wt \Sigma)$ (see Fig.~\ref{ff1}).

Denote with $\tau_{_{(1)}}$ and $\wt \tau_{_{(1)}}$ the corresponding Fredholm determinants of the operators def\/ined as described; following the same steps as above. Our goal is to show that
\begin{Theorem}
The ratio of the two Fredholm determinants is given by
\begin{gather*}
\wt \tau_{_{(1)}}={ \tau_{_{(1)}}} \prod_{v \in \operatorname{Int}(\wt \Sigma) \cap \operatorname{Ext}(\Sigma):\atop a(v)=0 } (c(v(\t);\t) )^{-{{\rm ord}}_{v}(a)}.
\end{gather*}
Note that the evaluation of $c$ at the zeroes of $a$ cannot vanish because $\det M \equiv 1$.
\end{Theorem}

\begin{proof} From the formula \eqref{theta} we get
\begin{gather*}
\pa \ln \frac { \wt \tau_{_{(1)}}}{ \tau_{_{(1)}}}
 = \left(\oint_{\wt \Sigma} - \oint_{\Sigma}\right) \operatorname{Tr} \big(\G_-^{-1} \G_-' \pa M M^{-1} \big) \frac {\d z}{2i\pi}\\
\hphantom{\pa \ln \frac { \wt \tau_{_{(1)}}}{ \tau_{_{(1)}}}=}{} +
 \left(\oint_{\wt \Sigma} - \oint_{\Sigma}\right) \left( \frac {(a)'}{a}(d \pa a - c \pa b )+ { ac'} \pa \left(\frac b a \right)\right)\frac {\d z}{2i\pi}.
\end{gather*}
Here $\G_-$ means the analytic extension of the solution $\G$ to the region $\operatorname{Ext}(\Sigma) \cup N(\Sigma)$. Since the integrand of the f\/irst term is holomorphic in the region bounded by $\wt \Sigma $ and $\Sigma$, it yields a~zero contribution by the Cauchy's theorem and we are left only with the second term, which is computable by the residue theorem;
\begin{gather*}
\pa \ln \frac { \wt \tau_{_{(1)}}}{ \tau_{_{(1)}}} =
\sum_{v \in \operatorname{Int}(\wt \Sigma) \cap \operatorname{Ext}(\Sigma):\atop a(v)=0 } \res{z=v}\left(
 \frac {(a)'}{a}(d \pa a - c \pa b ) + {c'}a \pa \left(\frac b a \right)\right).
\end{gather*}
We are assuming that $M(z;\t)$ is analytic and also that $\det M \equiv 1$.
Now note that a zero $v(\t)$ of $a(z;\t)$ in general depends on $\t$; suppose that $a(z;\t) = (z-v(\t))^k (C_k(\t) + \mathcal O(z-v(\t)))$; then
\begin{gather*}
\frac {\pa a(z;\t)}{a(z;\t)} =- \frac {k\pa v(\t)}{z-v(\t)} + \mathcal O(1),
\end{gather*}
and hence the residue evaluation at $z=v(\t)$ yields (we use $0 = a'\pa v + \pa a|_v$ and $\pa (ad)|_v = d \pa a|_v = \pa(bc)|_v$)
\begin{gather*}
\res{z=v}\left( \frac {(a)'}{a}\left( d \pa a - c \pa b \right)+ {ac'} \pa \left(\frac b a \right) \right)= k \left( {-}d a'\big|_v \pa v - c \pa b\big|_v\right) + k bc' \pa v\\
\qquad{} = k \left( d \pa a \big|_v - c \pa b\big|_v\right) + k bc' \pa v = k \left( \pa (bc) \big|_v - c \pa b\big|_v\right) + k bc' \pa v\\
\qquad{} = k b \pa c \big|_v + k bc'\big |_v \pa v \mathop{=}^{bc|_v =-1} - k \pa \ln c(v(\t);\t)
\end{gather*}
(note that $c(v(\t);\t)$ cannot be zero because $\det M\equiv 1$ and $z=v$ is already a zero of $a$).
\end{proof}

\subsubsection{Determinants for dif\/ferent factorizations}

The tau functions (Fredholm determinants) def\/ined thus far should be understood as def\/ining a section of a line bundle over the {\em loop group} space; this is the line bundle associated to the two fom $\delta \omega_M$~\eqref{deltaomega}.

This simply means that on the intersection of the open sets where the factorizations {(1)}, {(2)} in~\eqref{Factorizations} can be made, we have
\begin{gather*}
\delta \ln \frac {\tau_{_{(1)}}}{\tau _{_{(2)}}}= \delta\ln \big( \Upsilon_{(1,2)}\big),
\end{gather*}
where
\begin{gather*}
\delta\ln \big( \Upsilon_{(1,2)}\big) :=\delta \big( \theta_{_{(1)}} - \theta_{_{(2)}} \big).
\end{gather*}
After a short computation we obtain
\begin{gather*}
\delta \ln \big( \Upsilon_{(1,2)}\big) = \oint_{\Sigma} \left(\frac { c' \delta b - b' \delta c}{1 + b c}\right) \ddz=
\oint_\Sigma\left( \delta \ln b \frac {\d}{\d z} \ln (ad) - \delta \ln (ad) \frac {\d}{\d z} \ln b \right)\ddz.
\end{gather*}
Observe the last expression; in principle the functions $b(z)$, $a(z)d(z)$ may have nonzero index around $\Sigma$, but in any case the functions $\delta \ln b$ and $(\ln (ad))'$ are single-valued because the increments of the logarithms are integer multiples of $2i\pi$ and hence locally constant in the space of deformations. To write explicitly the transition function (or rather a representative of the same cocycle class) we need to choose a point $z_0\in \operatorname{Int}(\Sigma)$; we choose $z=0$ without loss of generality. Let $K =\operatorname{ind}_\Sigma b$, $L = \operatorname{ind}_{\Sigma} (ad)$; then we can rewrite the above expression as follows
\begin{gather*}
\delta \ln \big( \Upsilon_{(1,2)}\big)=
\oint_\Sigma \bigg[ \delta \ln \left(\frac b {z^K} \right) \frac {\d}{\d z} \ln \left(\frac {ad} {z^L} \right)-
\delta \ln \left(\frac {ad} {z^L} \right) \frac {\d}{\d z} \ln \left(\frac {b} {z^K} \right)\\
\hphantom{\delta \ln \big( \Upsilon_{(1,2)}\big)=\oint_\Sigma \bigg[}{}
- \delta \ln \left(\frac {ad} {z^L} \right) \frac K z + \delta \ln \left(\frac {b} {z^K} \right) \frac L z\bigg] \ddz,
\end{gather*}
and after integration by parts (which is now possible since all functions involved are single-valued in $N(\Sigma)$)
\begin{gather*}
\delta \ln \left( \Upsilon_{(1,2)}\right) =\delta \oint_\Sigma
\left[ \ln \left(\frac b {z^K} \right) \frac {\d}{\d z} \ln \left(\frac {ad} {z^L} \right)
-\ln \left(\frac {ad} {z^L} \right) \frac K z + \ln \left(\frac {b} {z^K} \right) \frac L z\right] \ddz,
\end{gather*}
so that the transition function admits the explicit expression
\begin{gather*}
\Upsilon_{(1,2)} =\exp \left[\oint_\Sigma \left[ \ln \left(\frac b {z^K} \right) \frac {\d}{\d z} \ln \left(\frac {ad} {z^L} \right)
- \ln \left(\frac {ad} {z^L} \right) \frac K z + \ln \left(\frac {b} {z^K} \right) \frac L z \right] \ddz\right].
\end{gather*}

\section{Conclusion}
We conclude this short note with a few comments.

First of all the choice of $\Sigma$ as a single closed contour is not necessary; we can have a disjoint union of closed contours or an unbounded contour as long as $M(z)$ converges to the identity suf\/f\/iciently fast as $|z|\to \infty$ or near the endpoints. The considerations extend with trivial modif\/ications. Our approach is similar in spirit to the approach used in \cite{GavrylenkoLisovyy} to express the tau function of a general isomonodromic system with Fuchsian singularities in terms of an appropriate Fredholm determinant.

Much less clear to the writer is how to handle the case where $\Sigma$ contains intersections; in this case we should stipulate a local ``no-monodromy'' condition at the intersection points as explained in \cite{BertolaIsoTau, BertolaCorrection}. The obstacle is not the issue of factorization but the fact that the resulting RHP of the IIKS type leads to an operator~$K$~\eqref{K} which is not of trace-class (and not even Hilbert--Schmidt, which would be suf\/f\/icient in order to construct a Hilbert--Carleman determinant).

Nonetheless, the two form $\delta \omega_M$ def\/ines a line bundle as explained and therefore it is possible to compute the ``transition functions'' of the line-bundle; this is precisely what is accomplished (in dif\/ferent setting and in dif\/ferent terminology) in recent works \cite{Its:2016kq, Its:2015qf} and the transition functions can be expressed in terms of explicit expressions, analogously to what we have shown here in this general but simplif\/ied setting.

\subsection*{Acknowledgements}
The author wishes to thank Oleg Lisovyy for asking a very pertinent question on the representation of the Malgrange form in terms of Fredholm determinants. Part of the thinking was done during the author's stay at the ``Centro di Ricerca Matematica Ennio de Giorgi'' at the Scuola Normale Superiore in Pisa, workshop on ``Asymptotic and computational aspects of complex dif\/ferential equations'' organized by G.~Filipuk, D.~Guzzetti and S.~Michalik. The author wishes to thank the organizers and the Institute for providing an opportunity of fruitful exchange.

\pdfbookmark[1]{References}{ref}
\LastPageEnding


\begin{thebibliography}{99}
\footnotesize\itemsep=0pt

\bibitem{BasorWidom-BO}
Basor E.L., Widom H., On a {T}oeplitz determinant identity of {B}orodin and
 {O}kounkov, \href{https://doi.org/10.1007/BF01192828}{\textit{Integral Equations Operator Theory}} \textbf{37} (2000),
 397--401, \href{https://arxiv.org/abs/math.FA/9909010}{math.FA/9909010}.

\bibitem{BertolaIsoTau}
Bertola M., The dependence on the monodromy data of the isomonodromic tau
 function, \href{https://doi.org/10.1007/s00220-009-0961-7}{\textit{Comm. Math. Phys.}} \textbf{294} (2010), 539--579,
 \href{https://arxiv.org/abs/0902.4716}{arXiv:0902.4716}.

\bibitem{BertolaCorrection}
Bertola M., Corrigendum: {T}he dependence on the monodromy data of the
 isomonodromic tau function, \href{https://arxiv.org/abs/1601.04790}{arXiv:1601.04790}.

\bibitem{BertolaCafasso1}
Bertola M., Cafasso M., The transition between the gap probabilities from the
 {P}earcey to the {A}iry process~-- a {R}iemann--{H}ilbert approach,
 \href{https://doi.org/10.1093/imrn/rnr066}{\textit{Int. Math. Res. Not.}} \textbf{2012} (2012), 1519--1568,
 \href{https://arxiv.org/abs/1005.4083}{arXiv:1005.4083}.

\bibitem{BorodinOkounkov}
Borodin A., Okounkov A., A {F}redholm determinant formula for {T}oeplitz
 determinants, \href{https://doi.org/10.1007/BF01192827}{\textit{Integral Equations Operator Theory}} \textbf{37} (2000),
 386--396, \href{https://arxiv.org/abs/math.CA/9907165}{math.CA/9907165}.

\bibitem{Chihara}
Chihara T.S., An introduction to orthogonal polynomials, \textit{Mathematics and
 its Applications}, Vol.~13, Gordon and
 Breach Science Publishers, New York~-- London~-- Paris, 1978.

\bibitem{ClanceyGohberg}
Clancey K.F., Gohberg I., Factorization of matrix functions and singular
 integral operators, \href{https://doi.org/10.1007/978-3-0348-5492-4}{\textit{Operator Theory: Advances and Applications}},
 Vol.~3, Birkh\"auser Verlag, Basel~-- Boston, Mass., 1981.

\bibitem{Gantmacher}
Gantmacher F.R., The theory of matrices. {V}ols.~1,~2, Chelsea Publishing Co.,
 New York, 1959.

\bibitem{GavrylenkoLisovyy}
Gavrylenko P., Lisovyy O., Fredholm determinant and Nekrasov sum
 representations of isomonodromic tau functions, \href{https://arxiv.org/abs/1608.00958}{arXiv:1608.00958}.

\bibitem{ItsHarnad}
Harnad J., Its A.R., Integrable {F}redholm operators and dual isomonodromic
 deformations, \href{https://doi.org/10.1007/s002200200614}{\textit{Comm. Math. Phys.}} \textbf{226} (2002), 497--530,
 \href{https://arxiv.org/abs/solv-int/9706002}{solv-int/9706002}.

\bibitem{Its:2016kq}
Its A., Lisovyy O., Prokhorov A., Monodromy dependence and connection formulae
 for isomonodromic tau functions, \href{https://arxiv.org/abs/1604.03082}{arXiv:1604.03082}.

\bibitem{Its:2015qf}
Its A., Lisovyy O., Tykhyy Yu., Connection problem for the
 sine-{G}ordon/{P}ainlev\'e~{III} tau function and irregular conformal blocks,
 \href{https://doi.org/10.1093/imrn/rnu209}{\textit{Int. Math. Res. Not.}} \textbf{2015} (2015), 8903--8924,
 \href{https://arxiv.org/abs/1403.1235}{arXiv:1403.1235}.

\bibitem{ItsIzerginKorepinSlavnov}
Its A.R., Izergin A.G., Korepin V.E., Slavnov N.A., Dif\/ferential equations for
 quantum correlation functions, \href{https://doi.org/10.1142/S0217979290000504}{\textit{Internat.~J. Modern Phys.~B}}
 \textbf{4} (1990), 1003--1037.

\bibitem{Malgrange:IsoDef1}
Malgrange B., Sur les d\'eformations isomonodromiques. {I}.~{S}ingularit\'es
 r\'eguli\`eres, in Mathematics and Physics ({P}aris, 1979/1982),
 \textit{Progr. Math.}, Vol.~37, Birkh\"auser Boston, Boston, MA, 1983,
 401--426.

\bibitem{Malgrange:Deformations}
Malgrange B., D\'eformations isomonodromiques, forme de {L}iouville,
 fonction~{$\tau$}, \href{https://doi.org/10.5802/aif.2052}{\textit{Ann. Inst. Fourier (Grenoble)}} \textbf{54} (2004),
 1371--1392.

\end{thebibliography}
\end{document}